\documentclass[11pt]{article}

\usepackage{amsmath, amsthm, amssymb}
\usepackage[hidelinks]{hyperref}
\usepackage{geometry}
\usepackage{authblk}
\usepackage{graphicx}
\usepackage{booktabs}

\newtheorem{theorem}{Theorem}
\newtheorem{proposition}{Proposition}
\newtheorem{lemma}{Lemma}
\newtheorem{corollary}{Corollary}
\newtheorem{definition}{Definition}
\newtheorem{remark}{Remark}

\newcommand{\bfa}{\mathbf{a}}

\newcommand{\intro}{\mathrm{intro}}
\newcommand{\pC}{p_C}

\title{Introspection Dynamics with Mutation in Additive Games}
\author[1]{Harry Foster}
\author[1,*]{Vince Knight}
\author[2]{Sebastian Krapohl}
\affil[1]{School of Mathematics, Cardiff University}
\affil[2]{Faculty of Social and Behavioural Sciences, University of Amsterdam}
\affil[*]{Corresponding author: \texttt{knightva@cardiff.ac.uk}}
\date{\today}

\begin{document}

\maketitle

\begin{abstract}
Cooperation in heterogeneous groups, where individuals differ in resources, productivity,
and behavioural responsiveness, underpins collective action across many social and
biological systems. Introspection dynamics, in which each player compares their payoff to
what they would have received under the alternative action, provides a natural learning
rule for such asymmetric settings. Couto and Pal showed that for \emph{additive}
games, those in which the payoff difference a player evaluates when considering a
switch is independent of the other players' actions, the stationary distribution of
introspection dynamics is a product measure. We extend this result to introspection
dynamics with \emph{mutation}, where a selected player switches to a random action
with some probability independent of payoffs, and with player-specific selection
intensities. We show that the product-measure structure is preserved, and we obtain
the explicit per-player cooperation probability
\(p_i=\phi_i(\delta_i)(1-\mu_{i0}-\mu_{i1})+\mu_{i0}\). As our main application we
consider the heterogeneous public goods game, where \(N\) players may differ in their
contributions \(\alpha_i\), public goods multipliers \(r_i\), and selection
intensities \(\beta_i\); the long-run cooperation probability admits the closed form
\[
  \pC = \frac{1}{N}\sum_{i=1}^{N}
        \left[\frac{1-\mu_{i0}-\mu_{i1}}{1+e^{\,\beta_i\alpha_i(1-r_i/N)}}+\mu_{i0}\right].
\]
Several structural consequences follow: a player-specific cooperation threshold at
\(r_i = N\) under symmetric mutation, a neutral-drift regime in which cooperation is
governed entirely by mutation bias, and a mutation-selection balance in which aggregate
cooperation is affine in the mutation rate, interpolating between the selection-driven
level and neutrality. Mutation also regularises the strong-selection limit, so the
closed form remains valid as \(\beta_i\to\infty\), where the mutation-free dynamics
degenerate.
\end{abstract}

\section{Introduction}

The public goods game is a canonical model of cooperation in evolutionary game theory:
players decide whether to contribute to a common resource, which is multiplied and
redistributed, creating a tension between individual incentive and collective
benefit~\cite{nowak2006five}. Classical analyses assume a homogeneous population, but
real groups are heterogeneous: players differ in their productivity, their contributions,
and how sensitively they respond to payoff differences~\cite{hauser2019social}.

A natural learning rule for such settings is \emph{introspection dynamics}, introduced
in~\cite{couto2022introspection}: at each time step a player
compares their current payoff to the payoff they \emph{would have received} had they
played the alternative action (keeping all other players fixed), and switches with a
probability governed by this difference. Because the comparison is
self-referential rather than imitative, introspection applies to arbitrary asymmetric
games where role-model imitation is unavailable or unnatural.

In~\cite{couto2023multiplayer} the authors extended introspection dynamics to multiplayer
asymmetric games. For a general game the transition rate for player \(i\) switching
depends on the actions of all other players through the payoff function, so the
stationary distribution must be computed by solving a linear system of dimension
\(2^N\). They identified an important exception: for \emph{additive} games, in which the
payoff difference a player evaluates when switching action is independent of the other
players' actions (a property also known as \emph{equal gains from switching}), the
stationary distribution factorises as a product measure over the players
(\cite{couto2023multiplayer}; Proposition~2), and for the linear public goods game it
admits a closed form (\cite{couto2023multiplayer}; Proposition~3).

The introspection model of~\cite{couto2022introspection,couto2023multiplayer} contains no
mutation: a selected player switches action solely according to the Fermi rule. Mutation,
a random action chosen independently of payoffs, is a standard ingredient of evolutionary
and learning dynamics~\cite{fudenberg2006evolutionary} and accounts for exploration,
error, and behavioural noise. In this
paper we incorporate per-player, per-action mutation into introspection dynamics and show
that the additive-game results of~\cite{couto2023multiplayer} persist: the stationary
distribution remains a product measure, and each player's long-run cooperation probability
has the explicit form \(p_i=\phi_i(\delta_i)(1-\mu_{i0}-\mu_{i1})+\mu_{i0}\). We also allow
player-specific selection intensities \(\beta_i\) rather than a single global value.
Specifically, we extend the product-measure result to mutation (Section~\ref{sec:decomp}),
record the structural consequences that hold for any additive game
(Section~\ref{sec:consequences}), and specialise them to the heterogeneous public goods
game (PGG), for which additivity holds and the closed-form cooperation probability follows
(Section~\ref{sec:payoff}).

The paper is organised as follows. Section~\ref{sec:setup} introduces the framework, the
additive-game condition, and introspection dynamics with mutation, noting that whilst the
donation-form Prisoner's Dilemma is additive, the Stag Hunt and the classical (RPST)
Prisoner's Dilemma are not.
Section~\ref{sec:decomp} shows that, with mutation, individual cooperation probabilities
are given by~\eqref{eq:pi} and the stationary distribution remains a product measure.
Section~\ref{sec:consequences} records the structural consequences that hold for any
additive game: a neutral-drift regime, a strong-selection limit, a cooperation threshold,
and a mutation-selection balance.
Section~\ref{sec:payoff} specialises these to the heterogeneous PGG, establishing
additivity and deriving the closed-form cooperation probability.
Section~\ref{sec:conclusion} concludes.

\section{Setup and notation}\label{sec:setup}

Consider \(N\) ordered players, each choosing from two actions \(\{C,D\}\) coded as
\(\{1,0\}\). The state space is \(S=\{0,1\}^N\) with \(|S|=2^N\). Each player \(i\)
has a payoff function \(f_i:S\to\mathbb{R}\) and player-specific selection intensity
\(\beta_i>0\). Writing \(\bfa_{i\to\bar{a}_i}\) for the state with player \(i\)'s
action flipped to \(\bar{a}_i=1-a_i\), the \emph{payoff difference} for player \(i\)
at state \(\bfa\) is

\[
  \Delta f_i(\bfa) := f_i(\bfa) - f_i(\bfa_{i\to\bar{a}_i}).
\]

Following~\cite{couto2023multiplayer}, a game is \emph{additive} for player \(i\) if
\(\Delta f_i(\bfa)\) depends only on \(a_i\), not on \(\bfa_{-i}\), the actions of all
individuals other than \(i\); equivalently, the payoff difference a player evaluates when
considering a switch is independent of the co-players' actions, a property also known as
\emph{equal gains from switching}. When this holds there is a constant
\(\delta_i\in\mathbb{R}\) such that

\begin{equation}\label{eq:delta}
  \Delta f_i(\bfa) = (1-2a_i)\,\delta_i,
\end{equation}

where \(\delta_i = \Delta f_i\!\big|_{a_i=0}\) is the payoff difference when player
\(i\) defects. When \(a_i=0\), \eqref{eq:delta} follows directly from the definition of
\(\delta_i\); when \(a_i=1\), \(\Delta f_i\!\big|_{a_i=1}=-\Delta f_i\!\big|_{a_i=0}=-\delta_i\),
giving \eqref{eq:delta}.

Note that not all games are \emph{additive}. For
example in~\cite{couto2022introspection} two of the two-player, two-strategy games
considered to illustrate introspection dynamics are the Prisoner's Dilemma \(M_1\) and
the Stag Hunt \(M_2\) of~\eqref{eq:games_from_couto_2022}, for \(0 < c_1, c_2 < b\).

\begin{equation}\label{eq:games_from_couto_2022}
    M_1 = \begin{pmatrix}
        b - c_1, b - c_2 & - c_1, b\\
        b, - c_2 & 0, 0\\
    \end{pmatrix}
    \qquad
    M_2 = \begin{pmatrix}
        b - c_1, b - c_2 & - c_1, 0\\
        0, - c_2 & 0, 0\\
    \end{pmatrix}
\end{equation}

For \(M_1\): \(\delta_1=c_1\) and \(\delta_2=c_2\) however for \(M_2\), \(\Delta
f_1((1, 1)) = b-c_1\) but \(\Delta f_1((1, 0)) = -c_1\). The Prisoner's Dilemma
(as defined by \(M_1\)) is \emph{additive} but the Stag Hunt game (as
defined by \(M_2\)) is not. Further games considered
in~\cite{couto2022introspection} are the Volunteer's dilemma and the Volunteer's
timing dilemma (a generalisation with more than 2 strategies) neither of which are
\emph{additive}.

More generally, the classical Prisoner's Dilemma is parameterised by the reward
\(R\), temptation \(T\), sucker payoff \(S\), and punishment \(P\) with
\(T > R > P > S\) and the social-efficiency condition \(2R > T + S\).
Additivity for player 1 requires the \emph{equality}
\[
  T - R \;=\; P - S \quad \text{(equivalently } P + R = S + T\text{)},
\]
that is, the gain from defecting against a cooperator must equal the gain from
defecting against a defector. The PD axioms are strict inequalities and do not
force this equality; a generic RPST Prisoner's Dilemma therefore fails it. For
instance, \(T=5\), \(R=3\), \(P=1\), \(S=0\) satisfies \(T>R>P>S\) and
\(2R=6>5=T+S\), yet \(T-R=2\neq 1=P-S\). The donation form \(M_1\) above is the
exceptional case where the equality holds by construction (both differentials
equal \(c\)).

To define introspection dynamics we introduce the following two quantities:

\begin{itemize}
    \item \textbf{Fermi function.} Each player \(i\) has a personal Fermi function
\(\phi_i(x)=(1+e^{\beta_i x})^{-1}\). Since \(\beta_i>0\), the exponential
\(e^{\beta_i x}\) is strictly increasing in \(x\), so \(\phi_i\) is strictly
decreasing with \(\phi_i(0)=\tfrac{1}{2}\) and \(\phi_i(x)+\phi_i(-x)=1\).
    \item \textbf{Mutation.} Each player \(i\) has mutation probabilities
\(\mu_{i0},\mu_{i1}>0\) with \(\mu_{i0}+\mu_{i1}<1\), where \(\mu_{i0}\) is the
probability of mutating to cooperation and \(\mu_{i1}\) is the probability of
mutating to defection, regardless of the intended choice of the dynamic.
\end{itemize}

At each step, one player \(i\) is selected
uniformly at random. Player \(i\) considers switching from action \(a_i\) to the
alternative \(\bar{a}_i=1-a_i\). The transition probability for the switch is

\begin{equation}\label{eq:intro-transition}
  T^{\intro}_{\bfa,\,\bfa_{i\to\bar{a}_i}}
  = \frac{1}{N}\bigl[(1-\mu_{i0}-\mu_{i1})\phi_i(\Delta f_i(\bfa))+\mu_{i,\bar{a}_i}\bigr],
\end{equation}

where \(\mu_{i,\bar{a}_i}\) denotes \(\mu_{i0}\) when
\(\bar{a}_i=C\) and \(\mu_{i1}\) when \(\bar{a}_i=D\). Since \(\phi_i\) is strictly
decreasing, a larger payoff gain from keeping the current action yields a smaller
switching probability. Setting \(\mu_{i0}=\mu_{i1}=0\) recovers the introspection
dynamics of~\cite{couto2022introspection,couto2023multiplayer}; the mutation terms add a
payoff-independent probability of switching to each action.

\begin{definition}[Cooperation probability]
Given an ergodic chain on \(S\) with stationary distribution \(\pi\), the
\emph{cooperation probability} is \(p_C=\pi\cdot s\), where
\(s_\bfa=\frac{1}{N}\sum_i a_i\) is the fraction of cooperators in state \(\bfa\).
\end{definition}

\section{Individual cooperation probabilities}\label{sec:decomp}

For additive games without mutation the stationary distribution of introspection
dynamics is a product measure (\cite{couto2023multiplayer}; Proposition~2). We show
that this structure persists with mutation, with a short direct
argument: for an additive game, \(\Delta f_i(\bfa)\) depends only on \(a_i\), so each time
player \(i\) is selected their next action is \(C\) with a fixed probability \(p_i\)
regardless of the current state. The marginal cooperation probability and the product form
then follow immediately.

\begin{theorem}[Individual cooperation probabilities]\label{thm:decomp}

For any additive game with constants \(\delta_i\), under introspection
dynamics with player-specific selection intensity \(\beta_i\) and mutation probabilities
\(\mu_{i0},\mu_{i1}>0\) with \(\mu_{i0}+\mu_{i1}<1\), the long-run probability that
player \(i\) cooperates is

\begin{align}
  p_i &= \phi_i(\delta_i)(1-\mu_{i0}-\mu_{i1})+\mu_{i0}. \label{eq:pi}
\end{align}
\end{theorem}

\begin{proof}

By additivity, \(\Delta f_i(\bfa)=(1-2a_i)\delta_i\) depends only on
    \(a_i\). When player \(i\) is selected, the probability their next action is
    \(C\) equals \(p_i\) regardless of their current action: if \(a_i=D\), the
    switch probability is
    \(p_i=\phi_i(\delta_i)(1-\mu_{i0}-\mu_{i1})+\mu_{i0}\); if \(a_i=C\), the
    probability of remaining at \(C\) is
    \(1-[\phi_i(-\delta_i)(1-\mu_{i0}-\mu_{i1})+\mu_{i1}]\), using
    \(\phi_i(x)+\phi_i(-x)=1\) this gives: \(1-[(1 -
    \phi_i(\delta_i))(1-\mu_{i0}-\mu_{i1})+\mu_{i1}]=p_i\). Since each selection
    of player \(i\) places them in state \(C\) with probability \(p_i\)
    independently of their current state, the long-run cooperation probability
    of player \(i\) is \(p_i\).

\end{proof}

Additivity makes a player's choice forgetful: each time player \(i\) is selected they
cooperate with the same probability \(p_i\), whatever the current state. This probability
is the payoff-driven Fermi value \(\phi_i(\delta_i)\), scaled by the chance of no mutation
and shifted by the chance of mutating to cooperation.

\begin{theorem}[Stationary distribution]\label{thm:stationary}

For any additive game with constants \(\delta_i\), the stationary
distribution of the introspection chain on $S=\{0,1\}^N$ is the product measure
\begin{equation}\label{eq:stationary}
  \pi_{\bfa} = \prod_{i=1}^{N} p_i^{\,a_i}(1-p_i)^{1-a_i}, \qquad \bfa\in S,
\end{equation}
where $p_i$ is given by~\eqref{eq:pi}.

\end{theorem}

\begin{proof}

By Theorem~\ref{thm:decomp}, each selection of player $i$ places them in state $C$
with probability $p_i$ independently of every other player's action. In stationarity
the players' actions are therefore independent, with player $i$ cooperating with
probability $p_i$ and defecting with probability $1-p_i$. The probability of state
$\bfa$ is the product of these marginals, giving~\eqref{eq:stationary}.

\end{proof}

In the long run the players' actions are therefore independent: the group behaves as
\(N\) separate biased coins, with player \(i\) showing cooperation with probability
\(p_i\). Taking the expectation of the cooperator fraction \(s\) under this product
measure, the cooperation probability is the mean of the individual probabilities,
\begin{equation}\label{eq:pC-mean}
  \pC = \frac{1}{N}\sum_{i=1}^{N} p_i.
\end{equation}

For \(\mu_{i0}=\mu_{i1}=0\), Theorem~\ref{thm:stationary} is Proposition~2
of~\cite{couto2023multiplayer}; the argument above shows that the product-measure
structure is unaffected by mutation, since mutation alters only the per-player marginal
\(p_i\) and not the independence between players.

\begin{remark}[The two-player donation game]
For the Prisoner's Dilemma \(M_1\) (\ref{eq:games_from_couto_2022})
of~\cite{couto2022introspection} the payoff \(f_1(\bfa)=-c_1 a_1+b a_2\) is additive
with \(\delta_1=c_1\), and similarly \(\delta_2=c_2\). Theorems~\ref{thm:decomp}
and~\ref{thm:stationary} give the per-player cooperation probability and the
product-measure stationary distribution
\[
  p_i = \frac{1-\mu_{i0}-\mu_{i1}}{1+e^{\beta_i c_i}}+\mu_{i0}, \qquad
  \pi_{\bfa} = \prod_{i=1}^{2} p_i^{\,a_i}(1-p_i)^{1-a_i}.
\]
Setting \(\mu_{i0}=\mu_{i1}=0\) and a common selection intensity \(\beta_i=\beta\)
recovers \(p_i=(1+e^{\beta c_i})^{-1}\) and
\[
  \pi = \frac{1}{(1+e^{\beta c_1})(1+e^{\beta c_2})}
        \bigl(1,\,e^{\beta c_2},\,e^{\beta c_1},\,e^{\beta(c_1+c_2)}\bigr)
  \quad\text{over } (CC,CD,DC,DD),
\]
the stationary distribution obtained by~\cite{couto2022introspection} for this game; they
observed that their two-player stationary distribution factorises precisely under the
additive condition, of which the donation game is an instance. Mutation preserves the
product form while shifting each marginal: with \(b=1\), \(c_1=0.6\), \(c_2=0.1\),
\(\beta=5\), and asymmetric mutation \(\mu_{i0}=0.05\), \(\mu_{i1}=0.15\) we obtain
\(p_1\approx0.088\) and \(p_2\approx0.352\), giving \(\pi_{DD}\approx0.591\),
\(\pi_{DC}\approx0.321\), \(\pi_{CD}\approx0.057\), and \(\pi_{CC}\approx0.031\).
\end{remark}

\section{Structural consequences}\label{sec:consequences}

The formula~\eqref{eq:pi} for the individual cooperation probability allows immediate
structural conclusions that hold for any additive game, before any specific payoff
structure is imposed. We record four: the two limits of the selection intensity, a
cooperation threshold, and the effect of mutation on aggregate cooperation. Without
mutation, several reduce to properties already observed
by~\cite{couto2023multiplayer}.

\begin{remark}[Neutral drift]\label{rem:neutral}

If \(\beta_i=0\) for all \(i\) then \(\phi_i\equiv\tfrac{1}{2}\), so the payoff parameters
\(\delta_i\) drop out of~\eqref{eq:pi} and cooperation is set entirely by mutation:
\[
  p_i = \frac{1+\mu_{i0}-\mu_{i1}}{2}, \qquad
  \pC = \frac{1}{2}+\frac{1}{2N}\sum_{i=1}^{N}(\mu_{i0}-\mu_{i1}),
\]
so \(\pC=\tfrac{1}{2}\) exactly when \(\sum_i(\mu_{i0}-\mu_{i1})=0\).

\end{remark}

A player insensitive to payoff differences ignores the underlying game: their long-run
behaviour is set by their mutation bias alone, and group cooperation reduces to the
average asymmetry between the mutation rates.

\begin{proposition}[Strong selection]\label{prop:strong}

For each player \(i\), \eqref{eq:pi} confines the cooperation probability to the open
interval \((\mu_{i0},\,1-\mu_{i1})\), and the endpoints are the strong-selection limits:
as \(\beta_i\to\infty\),
\[
  p_i \to \mu_{i0} \ \text{when}\ \delta_i>0, \qquad
  p_i \to 1-\mu_{i1} \ \text{when}\ \delta_i<0.
\]

\end{proposition}

\begin{proof}

Since \(0<\phi_i(\delta_i)<1\) and \(1-\mu_{i0}-\mu_{i1}>0\), \eqref{eq:pi} gives
\(\mu_{i0}<p_i<1-\mu_{i1}\). As \(\beta_i\to\infty\),
\(\phi_i(\delta_i)=(1+e^{\beta_i\delta_i})^{-1}\to0\) when \(\delta_i>0\) and \(\to1\)
when \(\delta_i<0\), so \(p_i\to\mu_{i0}\) or \(p_i\to1-\mu_{i1}\) respectively.

\end{proof}

A unique stationary distribution requires finite \(\beta\)~\cite{couto2023multiplayer}.
With mutation the stationary distribution is a non-degenerate product
measure even in this best-response limit, whereas without mutation the chain degenerates
to a point mass at the dominant-action profile. Mutation therefore extends the exact
analysis to arbitrarily strong selection.

\begin{corollary}[Cooperation threshold]\label{cor:threshold}

Player \(i\) cooperates with probability greater than \(\tfrac{1}{2}\) if and only if
\[
  \phi_i(\delta_i) > \frac{\tfrac{1}{2}-\mu_{i0}}{1-\mu_{i0}-\mu_{i1}}.
\]
When \(\mu_{i0}=\mu_{i1}\) this reduces to \(\delta_i<0\). A sufficient condition for
\(\pC>\tfrac{1}{2}\) is that the inequality holds for every \(i\).

\end{corollary}

\begin{proof}

Since \(1-\mu_{i0}-\mu_{i1}>0\),
\[
  p_i > \tfrac{1}{2}
  \iff (1-\mu_{i0}-\mu_{i1})\phi_i(\delta_i) > \tfrac{1}{2}-\mu_{i0}
  \iff \phi_i(\delta_i) > \frac{\tfrac{1}{2}-\mu_{i0}}{1-\mu_{i0}-\mu_{i1}}.
\]
When \(\mu_{i0}=\mu_{i1}\), the right-hand side equals \(\tfrac{1}{2}\), and since
\(\phi_i\) is strictly decreasing with \(\phi_i(0)=\tfrac{1}{2}\), this is equivalent
to \(\delta_i<0\). The claim about \(\pC\) follows by averaging~\eqref{eq:pC-mean}.

\end{proof}

With symmetric mutation a player cooperates the majority of the time precisely when the
payoff favours cooperation, \(\delta_i<0\); mutation does not move the threshold, it only
softens the transition across it. With asymmetric mutation the threshold can be met even
when the payoff favours defection. Substituting
\(\phi_i(\delta_i)=(1+e^{\beta_i\delta_i})^{-1}\), the condition \(p_i>\tfrac{1}{2}\) is
equivalent to
\[
  \beta_i\delta_i < \ln\frac{\tfrac{1}{2}-\mu_{i1}}{\tfrac{1}{2}-\mu_{i0}}.
\]
The right-hand side is positive precisely when the mutation is biased towards cooperation
(\(\mu_{i0}>\mu_{i1}\)). Thus when the payoff favours defection (\(\delta_i>0\)) a
sufficiently strong cooperative bias still yields \(p_i>\tfrac{1}{2}\); conversely, when
the payoff favours cooperation (\(\delta_i<0\)) a defection bias (\(\mu_{i1}>\mu_{i0}\))
can drive \(p_i<\tfrac{1}{2}\).

\begin{proposition}[Mutation-selection balance]\label{prop:balance}

Under common symmetric mutation \(\mu_{i0}=\mu_{i1}=\mu\) for all \(i\), with
\(0\le\mu\le\tfrac{1}{2}\), the aggregate cooperation probability is affine in \(\mu\),
\[
  \pC = (1-2\mu)\,\Phi + \mu, \qquad
  \Phi := \frac{1}{N}\sum_{i=1}^{N}\phi_i(\delta_i),
\]
so that \(\pC=\Phi\) at \(\mu=0\), \(\pC=\tfrac{1}{2}\) at \(\mu=\tfrac{1}{2}\), and
\[
  \frac{\partial\pC}{\partial\mu} = 1-2\Phi.
\]
Mutation raises aggregate cooperation when \(\Phi<\tfrac{1}{2}\) and lowers it when
\(\Phi>\tfrac{1}{2}\).

\end{proposition}

\begin{proof}

With \(\mu_{i0}=\mu_{i1}=\mu\), \eqref{eq:pi} gives
\(p_i=(1-2\mu)\phi_i(\delta_i)+\mu\). Averaging over \(i\) using~\eqref{eq:pC-mean} gives
the stated expression for \(\pC\); affinity in \(\mu\), the two endpoints, and the
derivative are immediate.

\end{proof}

Here \(\Phi\) is the cooperation level of the mutation-free dynamics
of~\cite{couto2023multiplayer}. Mutation interpolates linearly between this
selection-driven value and neutrality, reaching \(\pC=\tfrac{1}{2}\) as
\(\mu\to\tfrac{1}{2}\): noise erases the influence of the payoffs.

\section{The heterogeneous public goods game}\label{sec:payoff}

The \emph{heterogeneous public goods game}~\cite{hauser2019social} is played by \(N\)
players with player-specific contributions \(\alpha_1,\ldots,\alpha_N>0\) and public goods
multipliers \(r_1,\ldots,r_N>1\). The payoff to player \(i\) in state
\(\bfa=(a_1,\ldots,a_N)\) is
\begin{equation}\label{eq:payoff}
  f_i(\bfa) = \frac{1}{N}\sum_{j=1}^N r_j\alpha_j a_j - \alpha_i a_i.
\end{equation}
The first term is a public good shared equally by all players, in which each
contribution \(\alpha_j a_j\) is scaled by the contributor's multiplier \(r_j\); the
second is player \(i\)'s own contribution cost. We write
\(P(\bfa)=\frac{1}{N}\sum_j r_j\alpha_j a_j\) for the shared pool.

The linear public goods game is additive~\cite{couto2023multiplayer}, and the
heterogeneous version~\eqref{eq:payoff}, with player-specific multipliers \(r_i\), is
additive for the same reason.

\begin{lemma}[Additivity of the PGG]
\label{lem:payoff-identity}

The payoff~\eqref{eq:payoff} is additive for every player \(i\), with
\begin{equation}\label{eq:delta-pi}
  \delta_i = \alpha_i\!\left(1-\frac{r_i}{N}\right),
\end{equation}
i.e.\ \(\Delta f_i(\bfa)=(1-2a_i)\,\alpha_i(1-r_i/N)\), independent of \(\bfa_{-i}\).
\end{lemma}

\begin{proof}
When player \(i\) switches from \(a_i\) to \(\bar{a}_i=1-a_i\), only their own term in
the pool changes, so
\(P(\bfa_{i\to\bar{a}_i}) = P(\bfa)+\tfrac{1}{N}r_i\alpha_i(1-2a_i)\).
Expanding both payoffs:
\begin{align*}
  f_i(\bfa) &= P(\bfa)-\alpha_i a_i, \\
  f_i(\bfa_{i\to\bar{a}_i})
  &= P(\bfa)+\tfrac{1}{N}r_i\alpha_i(1-2a_i)-\alpha_i(1-a_i).
\end{align*}
Subtracting:
\begin{align*}
  \Delta f_i
  &= -\frac{r_i}{N}\alpha_i(1-2a_i) + \alpha_i(1-2a_i) \\
  &= \alpha_i(1-2a_i)\!\left(1-\frac{r_i}{N}\right).
\end{align*}
The pool \(P(\bfa)\) cancels entirely, confirming additivity with
\(\delta_i=\alpha_i(1-r_i/N)\).
\end{proof}

The incentive to switch therefore depends only on the player's own contribution
\(\alpha_i\) and on whether their multiplier exceeds the group size: \(\delta_i<0\),
an incentive to cooperate, precisely when \(r_i>N\). The other players' choices never
enter.

\begin{corollary}[Exact cooperation probability]\label{cor:exact-pC}

For the heterogeneous public goods game~\eqref{eq:payoff}, the long-run cooperation
probability is

\begin{equation}\label{eq:pC}
  \boxed{
  \pC = \frac{1}{N}\sum_{i=1}^{N}
        \left[\frac{1-\mu_{i0}-\mu_{i1}}{1+e^{\,\beta_i\alpha_i(1-r_i/N)}}+\mu_{i0}\right]
  }.
\end{equation}

\end{corollary}

\begin{proof}

By Lemma~\ref{lem:payoff-identity}, \(\delta_i=\alpha_i(1-r_i/N)\). Substituting
\(\phi_i(\delta_i)=(1+e^{\beta_i\alpha_i(1-r_i/N)})^{-1}\) into~\eqref{eq:pi} and
averaging by~\eqref{eq:pC-mean} gives~\eqref{eq:pC}.

\end{proof}

Aggregate cooperation is thus the average across players of an individual logistic
response to their own incentive \(\delta_i\), adjusted for mutation; no interaction
terms between players survive.

For \(\mu_{i0}=\mu_{i1}=0\) and a single global selection intensity, \eqref{eq:pC} reduces
to the stationary cooperation level of the linear public goods game in Proposition~3
of~\cite{couto2023multiplayer}.

Figure~\ref{fig:validation} validates Corollary~\ref{cor:exact-pC} against exact
values and simulation across group sizes.
Table~\ref{tab:stationary} verifies Theorem~\ref{thm:stationary}: for \(N=3\) with
per-player multipliers \(r_1=1\), \(r_2=3=N\), \(r_3=9\), contributions
\(\alpha_1=1\), \(\alpha_2=2\), \(\alpha_3=3\), \(\beta=2\), and \(\mu_{i0}=\mu_{i1}=0.1\),
the individual cooperation probabilities are \(p_1\approx0.267\), \(p_2=0.500\),
\(p_3\approx0.900\), and the formula and exact values of \(\pi_{\mathbf{a}}\)
agree to four decimal places for all eight states.

\begin{figure}[ht]
  \centering
  \includegraphics[width=\textwidth]{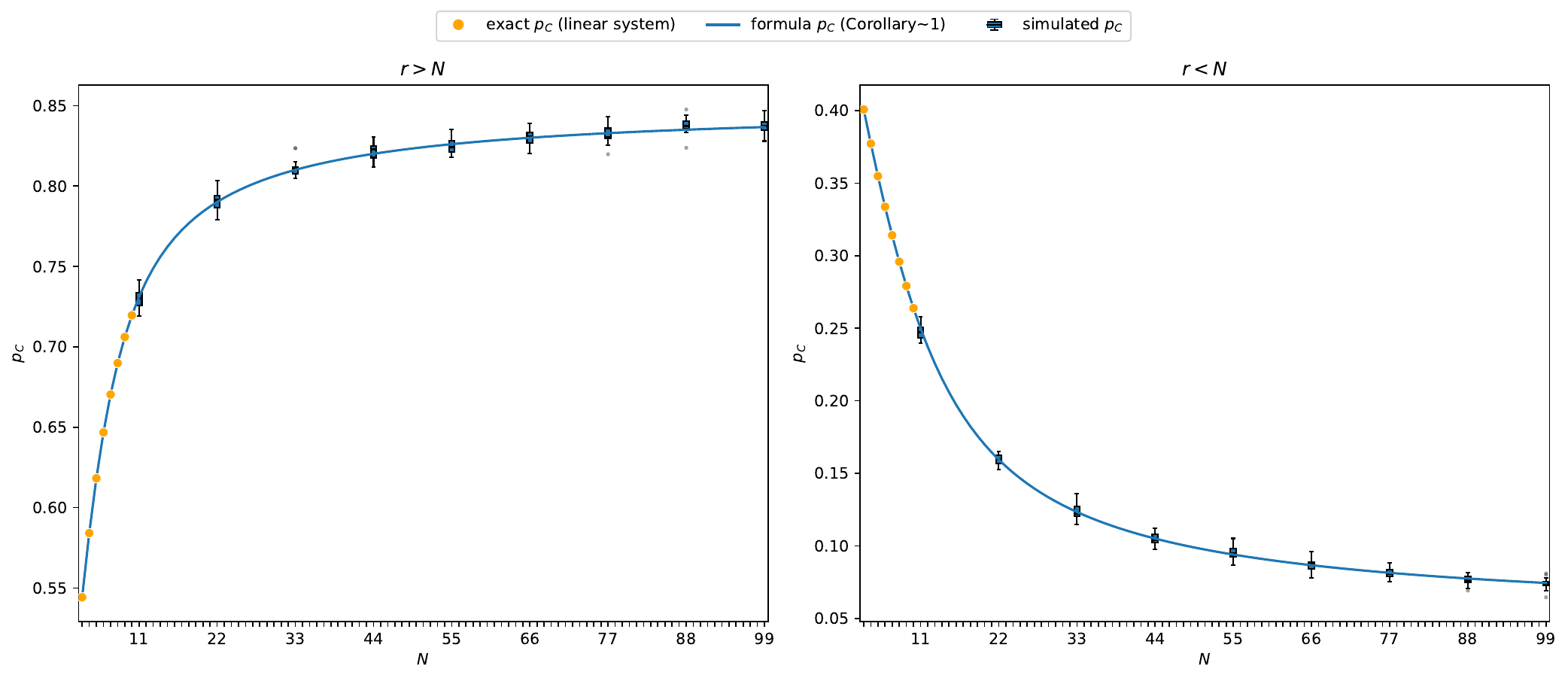}
  \caption{Validation of formula~\eqref{eq:pC} on a heterogeneous group with
    \(\alpha_i=i\), \(\beta=0.5\), \(\mu_{i0}=0.05\), \(\mu_{i1}=0.15\).
    \textit{Left:} \(r=2N\) (so \(r>N\) for all \(N\)).
    \textit{Right:} \(r=N/2\) (so \(r<N\) for all \(N\)).
    Orange dots: exact \(p_C\) from the full \(2^N\) linear system (feasible only
    for small \(N\)).
    Blue curve: formula~\eqref{eq:pC} (Corollary~\ref{cor:exact-pC}).
    Blue boxes: distribution of \(p_C\) across 19 independent runs of 5{,}000
    simulation steps each (with the first 500 steps discarded as warm-up;
    outliers shown).
    All values obtained using~\cite{harry_foster_2026_19495532}.}
  \label{fig:validation}
\end{figure}

\begin{table}[ht]
  \centering
  \begin{tabular}{lrr}
\toprule
State $\mathbf{a}$ & Formula $\pi_{\mathbf{a}}$ & Exact $\pi_{\mathbf{a}}$ \\
\midrule
DDD & 0.0367 & 0.0367 \\
DDC & 0.3299 & 0.3299 \\
DCD & 0.0367 & 0.0367 \\
DCC & 0.3299 & 0.3299 \\
CDD & 0.0133 & 0.0133 \\
CDC & 0.1201 & 0.1201 \\
CCD & 0.0133 & 0.0133 \\
CCC & 0.1201 & 0.1201 \\
\bottomrule
\end{tabular}

  \caption{Stationary-distribution probabilities \(\pi_{\mathbf{a}}\) for all
    eight states \(\mathbf{a}\in\{D,C\}^3\) with \(N=3\), \(\alpha_1=1\),
    \(\alpha_2=2\), \(\alpha_3=3\), \(\beta=2\), \(\mu_{i0}=\mu_{i1}=0.1\), and
    per-player multipliers
    \(r_1=1{<}N{=}r_2{=}3{<}r_3{=}9\).
    Formula: product measure~\eqref{eq:stationary} (Theorem~\ref{thm:stationary});
    Exact: values from the \(2^N\) linear system obtained
    using~\cite{harry_foster_2026_19495532}.
    Both columns to four decimal places.}
  \label{tab:stationary}
\end{table}

We now specialise the structural consequences of Section~\ref{sec:consequences} to the
public goods game, where \(\delta_i=\alpha_i(1-r_i/N)\) makes every quantity explicit in
the group size and the players' own parameters.

For the PGG, \(\delta_i<0\) precisely when \(r_i>N\), so
the symmetric-mutation threshold of Corollary~\ref{cor:threshold} becomes \(r_i>N\): a
player cooperates the majority of the time exactly when their personal multiplier on the
public good exceeds the group size, that is, when they extract more from the pool per unit
contributed than they put in. Without mutation this is the cooperation-dominance condition
of~\cite{couto2023multiplayer} (their \(c_i<b_i/N\)). The centre panel of
Figure~\ref{fig:panel} illustrates this, with curves at \(r>N\) lying above
\(p_i=\tfrac{1}{2}\) and curves at \(r<N\) lying below it. Since the comparison is between
\(r_i\) and \(N\), a player with a fixed multiplier becomes a non-benefactor once the
group exceeds \(r_i\): cooperation in large groups therefore requires multipliers that
grow with \(N\). For a non-beneficiary (\(r_i<N\), so \(\delta_i>0\)) a cooperative
mutation bias \(\mu_{i0}>\mu_{i1}\) can still raise \(p_i\) above \(\tfrac{1}{2}\), as in
Corollary~\ref{cor:threshold}.

At \(\beta_i=0\) the payoff parameters drop out
(Remark~\ref{rem:neutral}); the right panel of Figure~\ref{fig:panel} shows
\(p_i=\tfrac{1}{2}\) for every \(\alpha_i\) and \(r_i\), and the left panel shows the
curves converging to \(p_i=\tfrac{1}{2}\) at \(\beta_i=0\) before diverging in the
direction set by \(r_i\) relative to \(N\).

\begin{figure}[ht]
  \centering
  \includegraphics[width=\textwidth]{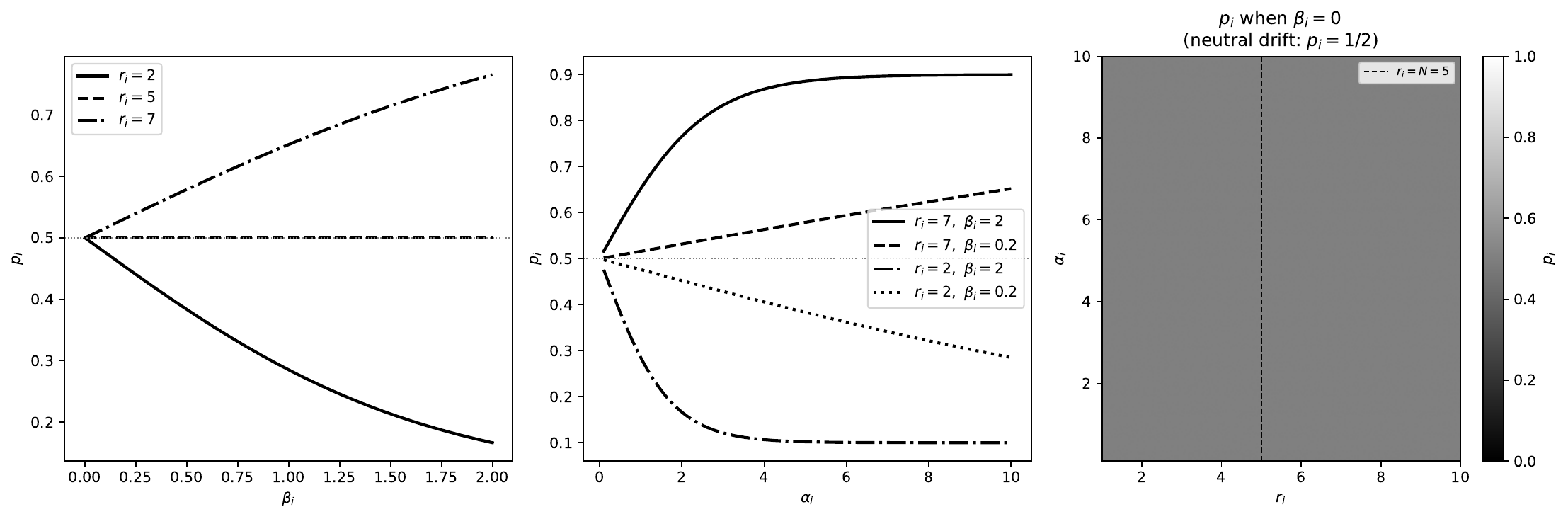}
  \caption{Structural properties of \(p_i\) for a single player with \(N = 5\),
    \(\mu_{i0} = \mu_{i1} = 0.1\).
    \textit{Left:} \(p_i\) as a function of \(\beta_i\) (\(\alpha_i = 2\))
    for \(r_i < N\), \(r_i = N\), and \(r_i > N\); all curves pass through
    \(p_i = \tfrac{1}{2}\) at \(\beta_i = 0\) (Remark~\ref{rem:neutral}).
    \textit{Centre:} \(p_i\) as a function of \(\alpha_i\) for two values of
    \(r_i\) and \(\beta_i\); curves above (below) \(\tfrac{1}{2}\) correspond
    to \(r_i > N\) (\(r_i < N\)).
    \textit{Right:} \(p_i(\alpha_i, r_i)\) at \(\beta_i = 0\), confirming
    \(p_i = \tfrac{1}{2}\) everywhere (Remark~\ref{rem:neutral}); the dashed
    line marks \(r_i = N\).}
  \label{fig:panel}
\end{figure}

For the PGG, \(\phi_i(\delta_i)>\tfrac{1}{2}\iff
r_i>N\), so the sign of \(\partial\pC/\partial\mu\) in Proposition~\ref{prop:balance} is
set by whether the group is, on average, beneficiaries: common mutation raises aggregate
cooperation when most players have \(r_i<N\) and lowers it when most have \(r_i>N\).

Figure~\ref{fig:mutation} isolates the
role of mutation for a single player. A larger selection intensity \(\beta_i\) makes the
player respond more sharply to payoff differences, adopting the higher-payoff action with
greater probability, whereas \(\beta_i=0\) is payoff-blind random choice. Without mutation
\(p_i\) is driven to \(0\) or \(1\) as \(\beta_i\) grows, whereas mutation confines
\(p_i\) to \([\mu_{i0},1-\mu_{i1}]\) (Proposition~\ref{prop:strong}). The two panels show
the same player (\(r_i=7\), \(\alpha_i=2\)) at \(N=5\) and \(N=200\): increasing the group
size alone turns the benefactor (\(r_i>N\), \(p_i\) rising to \(1-\mu\)) into a
non-benefactor (\(r_i<N\), \(p_i\) falling to \(\mu\)).

\begin{figure}[ht]
  \centering
  \includegraphics[width=\textwidth]{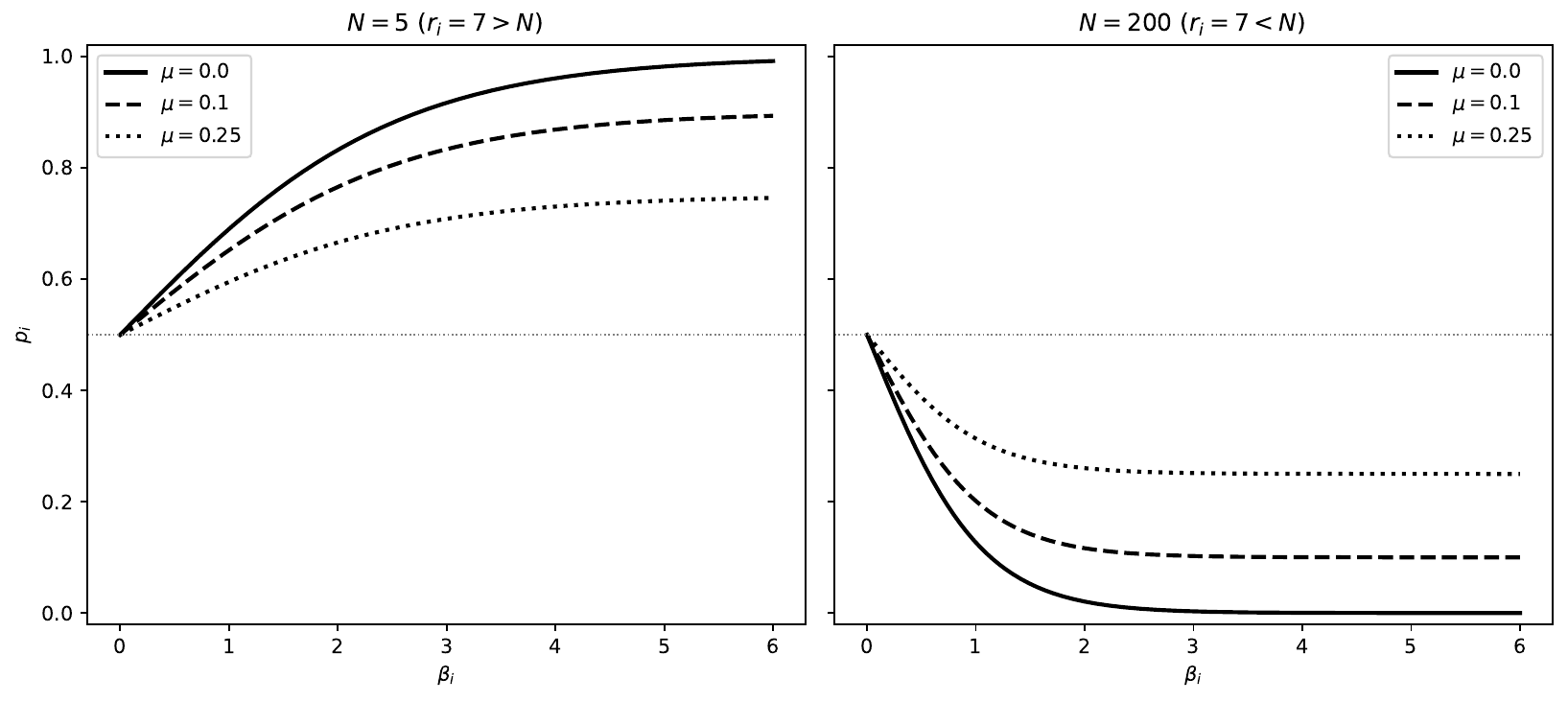}
  \caption{Effect of mutation on the cooperation probability of a single player with
    multiplier \(r_i=7\) and contribution \(\alpha_i=2\), as a function of the selection
    intensity \(\beta_i\) for symmetric mutation \(\mu\in\{0,0.1,0.25\}\), at two group
    sizes.
    \textit{Left:} \(N=5\), so \(r_i>N\) and the player is a benefactor; \(p_i\) rises
    towards \(1-\mu\) as \(\beta_i\) grows.
    \textit{Right:} \(N=200\), so the same player now has \(r_i<N\) and is a
    non-benefactor; \(p_i\) falls towards \(\mu\).
    In both panels the \(\mu=0\) curve is the mutation-free dynamics
    of~\cite{couto2023multiplayer}, reaching \(1\) or \(0\); mutation confines \(p_i\) to
    \([\mu,1-\mu]\). Increasing the group size alone turns the benefactor into a
    non-benefactor.}
  \label{fig:mutation}
\end{figure}

\begin{remark}[Role of heterogeneity]\label{rem:heterogeneity}

For the PGG with \(r_i\neq N\), both \(\partial p_i/\partial\alpha_i\) and
\(\partial p_i/\partial\beta_i\) are positive when \(r_i>N\) and negative when \(r_i<N\):
each derivative carries the factor \((1-r_i/N)\) multiplied by the negative quantity
\(\phi_i'\), so its sign is that of \(r_i/N-1\). Indeed, \(\beta_i\) and \(\alpha_i\) enter
\(p_i\) only through the product \(\beta_i\alpha_i\), so a change in selection intensity has
the same effect as a proportional change in contribution. Under symmetric mutation, raising
the mutation rate instead pulls \(p_i\) towards \(\tfrac{1}{2}\), lowering cooperation when
\(r_i>N\) and raising it when \(r_i<N\); with asymmetric mutation \(p_i\) is pulled towards
its neutral value \((1+\mu_{i0}-\mu_{i1})/2\).

\end{remark}

A player who benefits from the public good (\(r_i>N\)) cooperates more as they contribute
more or respond more strongly to payoff differences; a non-benefactor (\(r_i<N\)) does the
opposite. Each player's long-run cooperation is thus determined entirely by their own
\((\alpha_i,\beta_i,r_i,\mu_{i0},\mu_{i1})\) and the group size \(N\); the other players'
parameters do not enter.

\section{Conclusion}\label{sec:conclusion}

Building on the additive-game results of~\cite{couto2023multiplayer}, we have shown that
introspection dynamics with mutation on any additive game remains exactly solvable. When
\(\Delta f_i\) depends only on \(a_i\), each selection of player \(i\) places them in state
\(C\) with probability \(p_i\) regardless of current state (Theorem~\ref{thm:decomp}), and
the stationary distribution is a product measure (Theorem~\ref{thm:stationary}); the
product-measure structure, due to~\cite{couto2023multiplayer} for the no-mutation case, is
preserved under mutation. For the heterogeneous public
goods game, the linear payoff structure implies additivity
(Lemma~\ref{lem:payoff-identity}), and the long-run cooperation probability has a closed
form with \(N\) terms rather than \(2^N\) (Corollary~\ref{cor:exact-pC}).

The approach does not extend to games that are not additive. When
\(\Delta f_i(\bfa)\) depends on \(\bfa_{-i}\), the per-player chains are coupled and
the stationary distribution is no longer a product measure; the full \(2^N\) linear
system must be solved in general. Games such as the Stag Hunt (\(M_2\) of
(\ref{eq:games_from_couto_2022})) contain payoff cross
terms \(a_i a_{-i}\) that cannot be eliminated by rescaling, so no reformulation
brings them within the scope of the present results.

The source code, figures, and data for this paper are under version control
and available at~\cite{foster2026repo}, in line with best practices for
research software~\cite{wilson2017good}.

\section{Acknowledgements}
We thank Marta C. Couto and Saptarshi Pal for pointing out that the product-measure
decomposition for additive games was already established in~\cite{couto2023multiplayer},
which an earlier version of this paper had failed to acknowledge.
Harry Foster's research was supported by EPSRC grant EP/Z535126/1.

\bibliographystyle{plain}
\bibliography{bibliography}

\end{document}